%% file: Stability_in_Conductivity_Imaging.tex
\documentclass[12pt]{amsart}
\usepackage{amsmath} 					
\usepackage{amsfonts} 					
\usepackage{amsthm}						
\usepackage{amssymb} 					
\usepackage{cite}						
\usepackage{graphicx} 					
\usepackage{epstopdf}					
\usepackage[english]{babel}
\usepackage{upgreek} 					
\usepackage[raggedright]{sidecap}		
\usepackage{comment}

\usepackage{caption} 					
\usepackage{import} 					

\usepackage[usenames,dvipsnames]{xcolor}
\definecolor{linkcolour}{rgb}{0,0.2,0.8}
\definecolor{refcolor}{rgb}{0.1,0.6,0.1}
\usepackage[pagebackref]{hyperref} 		
\hypersetup{
	colorlinks,urlcolor=linkcolour,			
	linkcolor=linkcolour,
	citecolor=refcolor,
	linktocpage}

\newtheorem{theorem}{Theorem}[section]
\newtheorem{proposition}{Proposition}[section]
\newtheorem{lemma}{Lemma}[section]
\newtheorem{definition}{Definition}[section]

\theoremstyle{remark}

\newtheorem{example}{Example}

\renewcommand{\v}[1]{\ensuremath{\mathbf{#1}}} 	

\newcommand{\R}{\mathbb{R}}					

\newcommand{\I}{\mathcal{I}}

\renewcommand{\S}{\mathcal{S}}

\renewcommand{\d}{\mathrm{d}} 				

\newcommand{\tu}{\tilde{u}}
\newcommand{\tsigma}{\tilde{\sigma}}
\newcommand{\grad}{\nabla}
\newcommand{\J}{\mathbf{J}}
\newcommand{\B}{\mathbf{B}}

\graphicspath{{./figures/}} 

\title[Stability in Conductivity Imaging]{Stability in Conductivity Imaging from partial measurements of one interior current}

\begin{document}

\author{Carlos Montalto}
\address{Department of Mathematics, University of Washington, Seattle, WA 98195-4350, USA}
\email{montcruz@math.washington.edu}

\author{Alexandru Tamasan}\thanks{A. Tamasan was supported by the NSF Grant DMS 1312883.}

\address{Department of Mathematics, University of Central Florida, Orlando, Florida 32816}
\email{tamasan@math.ucf.edu}

\subjclass[2000]{35R30, 35J65, 65N21}

\date{}

\keywords{hybrid inverse problems, current density imaging, magnetic resonance, electrical impedance tomography, 1-Laplacian}

\begin{abstract}
We prove a stability result in the hybrid inverse problem of  recovering the electrical conductivity from partial
knowledge of one current density field generated inside a body by an imposed boundary voltage.
The region where interior data stably reconstructs the conductivity is well defined by a combination of the
exact and perturbed data.
\end{abstract}

\maketitle

\section{Introduction}

We consider the problem of recovering the electrical conductivity of a body from partial interior measurements of one current density field generated by an imposed boundary voltage. Originally, one sought to recover the conductivity solely from boundary measurements. Most generally formulated by Calder\'on \cite{alberto2006inverse}, the problem  has been extensively studied: the unique determination results  \cite{SylvesterUhlmann1987, Nachman1988,  Nachman1996, Bukhgeim2008,Astala-Paivatina:2006} are  known to
have weak (of logarithmic type) stability  estimates  \cite{Alessandrini1988}, which cannot be, in general, improved \cite{mandache2001exponential}. The need for determining the conductivity by more robust methods has lead to considering new problems in the conductivity imaging, which employ some interior data. Such interior knowledge is obtained by  coupling the electromagnetic model with the physics of some high resolution imaging method. For example,  magnetic resonance measurements are employed in  \cite{KwonLeeYoon2002:Equipotential, KwonWooYoonSeo2002:MREIT} to determine (one component) of the magnetic field inside, or in   \cite{zhang1992electrical,HasanovMaYoonNachmannJoy2004:NewApproachCDII, Lee2004:ReconstructionMREIT} to obtain the current  density field inside, whereas ultrasound measurements are employed in \cite{GebauerScherzer2008:UMEIT,Capdeboscq-Fehrenbach-Gournay2009imaging,Bal-Bonnetier-Monard-Faouzi2011inverse} to obtain the interior knowledge of power densities; see \cite{nachman-tamasan-timonov2011-CDII, Seo-Woo2011MREIT, Scherzer-Widlak:2012, Tamasan-Timonov2014} for some recent reviews.

This work concerns the problem of stability in determining the conductivity from partial knowledge of one current density inside. The interior data can be obtained from Magnetic Resonance measurements as discovered  in \cite{ScottJoy1991}. A first attempt to image the conductivity from the interior data of current density fields goes back to \cite{zhang1992electrical}. Since then several mathematical methods have been proposed, see  \cite{kim2002nonlinear, KwonLeeYoon2002:Equipotential, Nachman-Tamasan-Timonov2007conductivity, Nachman-Tamasan-Timonov2009MRIEIT, nachman2010reconstruction} for the isotropic case, and \cite{monard2012inverse, monard2013inverse, Bal-Guo-Monard:2014} for the anisotropic case.

Of  interest here, the work in \cite{Nachman-Tamasan-Timonov2007conductivity, Nachman-Tamasan-Timonov2009MRIEIT} shows that interior knowledge of the magnitude of one current density field generated by an imposed boundary voltage uniquely determine the conductivity inside. A general approach in \cite{Kuchment-Steinhauer2012} shows that using multiple measurements (two in the planar case) the linearized problem is stable, up to a finite dimensional kernel. Stability for the linear and non-linear problem from knowledge of the magnitude of one current density field was proved in \cite{montalto2013stability}. The second author would like to point out a small omission in \cite[Theorem 1.2]{montalto2013stability}, were the assumption that the voltage potentials must be equal at the boundary needs to be added. This is a natural assumption and it was clearly used in Equation (21) in there, but unfortunately omitted in the final version of the paper.

When the entire current density field is known,  a stability result has been obtained in  \cite{Kim-Lee:2014}. For the anisotropic conductivities  we mention the work in \cite{Bal-Guo-Monard:2014} which a proved stability  in two dimensions from full knowledge of four current density fields. All these stability results assume perturbations within the range of the data, locally near a given conductivity. If one assumes general perturbations in the interior data, the recent uniqueness result in \cite{Moradifam-Nachman-Tamasan2014:Uniqueness} combined with the structural stability in \cite{Nashed-Tamasan2010:Structural} shows that the voltage potential can still be stably recovered inside.  However,  the class of smoothness is insufficient to further yield stability for the conductivity.

In the case of partial interior data, unique determination of planar conductivity have been showed to be possible in certain subdomains (see the \emph{injectivity region} defined below) in \cite{nachman-tamasan-timonov2011-CDII}, but the stability question has been left open. While the voltage potential  can be stably recovered in certain subregions from partial interior data the  stability estimates obtained in \cite{VerasTamasanTimonov2015} do not guarantee the necessary regularity to extend to estimates for the conductivity.

In this paper we assume that knowledge of the current density field  is only partially available inside.  In fact we will show that knowledge of a certain component of the field suffices. We identify the specific subregions where conductivity can be stably recovered from partial interior and boundary data. The boundary voltage potential imposed to generate the current density field is assumed unperturbed, but need not be known everywhere. In general the stability region is a subset of the injectivity region as illustrated in Figure \ref{fig:visible-region-non-connected}. However, if the accessible part of the boundary  (where the voltage potential  is known)  is simply connected, then the injectivity region and the stability region coincide under some geometric (strict convexity) assumptions on the shape of the boundary.

Different from the work in \cite{montalto2013stability} we use a non-linear version of the decomposition of the data operator and identify a natural component of the current density field sufficient  to yield stability.  To authors' knowledge, this is the first result where either injectivity or stability is established from knowledge of just one component of the current density. The required component depends both on the measured and on the exact data.

Similar to all the works  in the hybrids methods, we also work with voltages free of singular points inside the domain. In dimension two, the condition of no critical points can be satisfied under the assumption that the boundary illumination is almost two-to-one \cite{Alessandrini1987:Critical, Nachman-Tamasan-Timonov2007conductivity}. In dimensions $n\geq 3$, it is still unclear if  similar conditions exists; the recent results in \cite{Honda-McLaughlin-Nakamura:2013, Alessandrini:2014} describing the set of singular points assume nonzero frequencies.

\section{Statement of the results}

Let $\Omega\subset\R^n$ be a bounded domain $C^{2,\alpha}$-diffeomorphic with the unit ball, $0<\alpha<1$. This assumption can be relaxed as explained in the proofs. Let $\sigma$ be a positive function in $C^{2,\alpha}(\overline{\Omega})$ and let $f\in C^{2,\alpha}({\partial \Omega)}$. From Schauder's regularity theory we know that the unique solution $u$ of the Dirichlet problem
\begin{align}
\label{Dirchlet-Problem}
\grad \cdot \sigma \grad u =0 \quad \mbox{in } \Omega, \quad
\qquad u|_{\partial \Omega} = f. \end{align}
is in $ C^{2,\alpha}(\overline{\Omega})$, see e.g. \cite[Theorem 6.18]{gt-2001}. We refer to such a $u$ as being $\sigma$-harmonic.  The corresponding {current density} vector field  $\J:C^{1,\alpha}(\overline{\Omega}) \to C^{1,\alpha}(\bar \Omega)$ is defined by

\begin{equation} \label{eq:current-density}
\J(\sigma) = - \sigma\nabla u.
\end{equation}

In order to state our partial data results we introduce the following notation.
For some $u \in C^1(\overline{\Omega})$ arbitrary (not necessarily $\sigma$-harmonic) with $|\grad u| > 0$ in $\overline{\Omega}$,  and $p$ an  arbitrary point in $\Omega$, let $\gamma_p$ denote the integral curve starting at $p$ in the direction of $\grad u$. Since $|\grad u| > 0$, there exist a first time, denoted by $t^+_p$ (resp. $t_p^-$), such that the integral curve starting at $p$ and moving in the direction  $\pm\grad u$  hits the boundary. We denote by
\[
l_p^\pm = \{\gamma_p(t): t\in[0,t_p^\pm] \}
\]
the segment of the integral curve from $t=0$ to $t=t_p^\pm$. Also, let
\[
\Sigma_p = \{ x \in \overline{\Omega}: u(x) = u(p)\}
\]
denote the level curve of $u$ passing through $p$, see Figure \ref{fig:sigma-p-and-lp}.

\begin{figure}[h] 
	\centering
 	\def\svgwidth{7cm}
	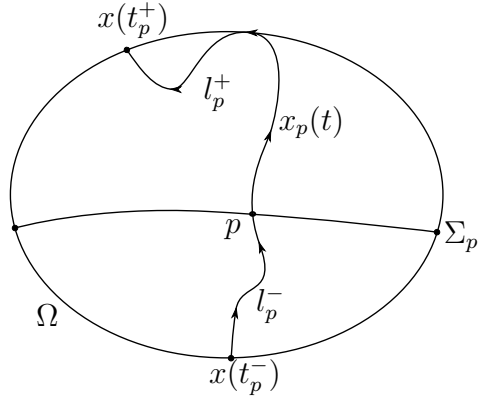	\caption{Illustration of the segment $l_p$ and the surface $\Sigma_p$.}
	\label{fig:sigma-p-and-lp}
\end{figure}

Let  $\Gamma$ be an open subset of the boundary to denote the accessible part. The following definitions specify the  interior subdomains, where the unique and stable determination of the conductivity can be guaranteed.
\begin{definition} Let $\Omega \subset \R^n$ be a bounded domain $C^{2,\alpha}$-diffeomorphic to the unit ball and let $u \in C^{2,\alpha}(\overline{\Omega})$, where $0<\alpha<1$. 
 \label{def:visible-region}
\begin{description}
\item[Injectivity region] We say that a point $p\in\overline\Omega$ is \emph{visible from $\Gamma$ along the equipotential set}  if
\[
\Sigma_p \cap \partial \Omega \subset \Gamma.
\]
The set $\I(\Gamma,u)$ of points in $\overline{\Omega}$ that are visible from $\Gamma$ along equipontential sets is called the \emph{injectivity region},
\begin{equation}
{\mathcal I}(\Gamma,u) := \{p \in \overline{\Omega}: \Sigma_p \cap \partial \Omega \subset \Gamma \}.
\end{equation}
\item[Stability region] We say that the trajectory through $p$  along  $\nabla u$ \emph{is visible from $\Gamma$ } if either
$
l^+_p \subset \I(\Gamma,u) \mbox{ or } l^-_p \subset \I(\Gamma,u).
$ The set $\S(\Gamma,u)$  of points in $\overline{\Omega}$ for which the corresponding trayectories  along $\nabla u$ are visible from $\Gamma$ is called the \emph{stability region} , i.e.,
\begin{equation}
{\mathcal S}(\Gamma, u) := \{p \in\I(\Gamma,u):  l^+_p \subset \I(\Gamma,u) \mbox{ or } l^-_p \subset \I(\Gamma,u) \}.
\end{equation}
\end{description}
\end{definition}

Clearly $\S(\Gamma,u) \subseteq \I(\Gamma,u)$ and, if $\Gamma$ is connected, the equality can hold as illustrated in Figure \ref{fig:stability-region}. In Figure \ref{fig:visible-region-non-connected} we illustrate the injectivity and stability regions, when $\Gamma$ is not connected.
\begin{figure}[h!] 
	\centering
 	\def\svgwidth{8cm}
 	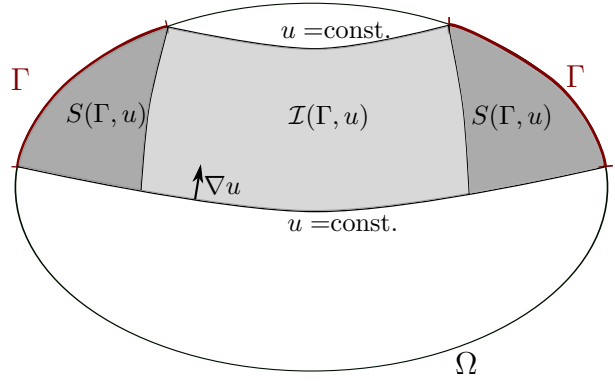
 	\caption{The injectivity region $\I(\Gamma,u)$ is the light grey region that contains the stability region $\S(\Gamma,u)$ in dark grey. }
  	\label{fig:visible-region-non-connected}
\end{figure}

\begin{figure}[h] 
	\centering
 	\def\svgwidth{7.5cm}
 	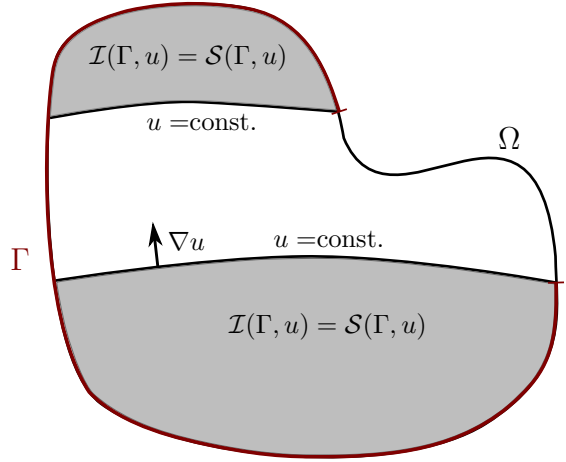
 	\caption{When $\Gamma$ is connected the visible region and the trajectory region can be the same.}
  	\label{fig:stability-region}
\end{figure}

Let $\tilde u$ be the voltage potential corresponding to some perturbed conductivity $\tilde\sigma\in C^{1,\alpha}(\overline{\Omega})$,
\begin{align}
\label{Dirchlet-Problem_perturbed}
\grad \cdot \tilde\sigma \grad \tilde u =0 \quad \mbox{in } \Omega, \quad
\qquad \tilde u|_{\partial \Omega} = f, \end{align}
and let $\J(\tilde\sigma) := - \tilde\sigma\nabla \tilde u$ be the corresponding {current density field}. Let
\begin{equation}\label{deltas}
\delta\sigma:=\sigma-\tsigma\quad\mbox{and}\quad\delta\J:=\J-\tilde\J
\end{equation}denote the corresponding perturbations.


We prove that $\delta\sigma$ is controlled by the component of $\delta\J$ in the direction $\nabla (u+\tu$) throughout $\S(\Gamma, u+\tu)$. To the author's knowledge the  result below is the first stability result from partial interior data. It is also the first time the recovery is done from just one component of the current density field.


For a vector field $\v w_0$ in $\Omega$  we denote the orthogonal projection onto $\v w_0$  by  $\Pi_{\v w_0}$, and onto the orthogonal complement (both in the Euclidean metric) by $\Pi^\perp_{\v w_0} := \mbox{Id } - \Pi_{\v w_0} $ the orthogonal projection; i.e., for an arbitrary vector field ${\v w}$,
\begin{equation}
\Pi_{\v w_0} \v w = \frac{\v w_0 \cdot \v w}{|\v w_0|^2} \v w_0 \quad \mbox{ and } \quad \Pi^\perp_{\v w_0}{\v w} = {\v w} - \Pi_{\v w_0}{\v w}.
\end{equation}

\begin{figure}[h] 
	\centering
 	\def\svgwidth{5.5cm}
 	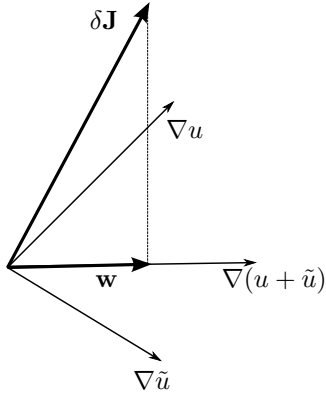
 	\caption{From knowledge to the scalar component of $\v w = \Pi_{\grad (u + \tu)}(\delta \J) $ we can stably recover the difference in the conductivities $\sigma - \tilde{\sigma}$. Here $\delta \J$ represents the difference $\J(\sigma) - \J(\tsigma)$.}
  	\label{fig:projection}
\end{figure}

\begin{theorem}[Stability with Partial Data] \label{theor:global-stability-CDII}
 Let $\Omega\subset\R^n$ be $C^{2,\alpha}$-diffeomorphic to the unit ball, and $\sigma, \tsigma \in C^{1,\alpha}(\overline{\Omega})$ be positive on $\overline{\Omega}$, for some $0<\alpha<1$. Let $u$ (respectively $\tu$) in $C^{2,\alpha}(\overline{\Omega})$ be $\sigma$(respectively $\tsigma$)-harmonic.  Let  $\Gamma\subseteq \partial \Omega$ be the union of finitely many open connected components, and  $\Gamma '\subset\subset\Gamma $ be a open subset compactly contained in $\Gamma$. Assume that
\begin{equation} \label{eq:bound-cond-sigma-and-u}
\sigma|_{\Gamma} = \tsigma|_{\Gamma}, \; u |_\Gamma = \tu|_\Gamma,
\end{equation}
and
\begin{equation}\label{nonsingularity}
|\grad (u + \tu)|>0,\quad\text{in} \quad \overline{\I(\Gamma', u+\tu)}.
\end{equation}
Then there exist $C> 0$ dependent on a lower bound ot  $|\nabla (u+ \tu)|$ in $\overline{\I(\Gamma', u+\tu)}$, the domain $\Omega$, and  the $C^1(\overline\Omega)$- norm of $\sigma$ and $\tsigma$, such that
\begin{equation}\label{J-Projection-Estimate}
\| \sigma - \tsigma \|_{L^2(\S(\Gamma', u+\tu))} \leq C \left|\left| \grad \cdot(\Pi_{\grad (u + \tu)} (\J(\sigma) - \J(\tsigma)) ) \right|\right|^{\frac{\alpha}{2+\alpha}}_{L^2(\I(\Gamma', u+\tu))}.
\end{equation}
In particular,
\begin{equation}\label{J-Scalar-Projection-Estimate}
\| \sigma - \tsigma \|_{L^2(\S(\Gamma', u+\tu))} \leq C \left|\left| \Pi_{\grad (u + \tu)}( \J(\sigma) - \J(\tsigma) )\right|\right|^{\frac{\alpha}{2+\alpha}}_{H^1(\I(\Gamma', u+\tu))}.
\end{equation}
\end{theorem}
The proof of the Theorem \ref{theor:global-stability-CDII} is presented in Section \ref{sec:proofofMainResult}.

A few remarks are in order. Note first, that the estimates \eqref{J-Projection-Estimate} and \eqref{J-Scalar-Projection-Estimate} concern norms over (possibly) different  domains. However, in many interesting situation $\S(\Gamma', u+\tu) = \I(\Gamma', u+\tu)$ as illustrated in Figure \nolinebreak \ref{fig:stability-region}.
In particular,  in the full data case with $\Gamma = \partial \Omega$, we have $\S(\Gamma',u +\tu) = \I(\Gamma', u+\tu) = \Omega$, and the estimate
\eqref{J-Scalar-Projection-Estimate} yields

\begin{equation}\label{eq:Stability-Whole-Current-Density}
\| \sigma - \tsigma \|_{L^2(\Omega)} \leq C \|\J(\sigma) - \J(\tsigma) \||^{\frac{\alpha}{2+\alpha}}_{H^1(\Omega)}.
\end{equation}

The stability estimate requires knowledge of a component of $\J$ in a direction $\grad(u+\tu)$ that is well defined by the exact and perturbed data due to the uniqueness results in \cite{Nachman-Tamasan-Timonov2009MRIEIT}. However, if $\tsigma$ is a priori close to $\sigma$ then the level curves of $\tu$ will be close to the level curves $u$, and it will suffice to project  $\delta\J$ onto the $\grad u$ with a penalty term that will depend on the apriori closeness assumption, as illustrated in Figure \ref{fig:controlled-stability-from-boundary}. More precisely, we have the following local stability result.


\begin{figure}[h] 
	\centering
 	\def\svgwidth{9cm}
 	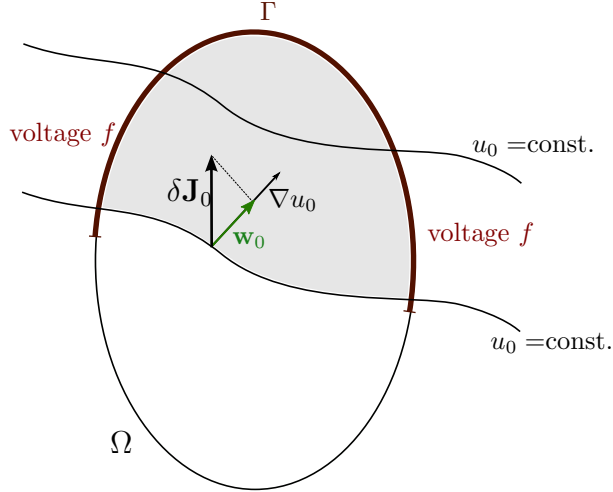
 	\caption{We illustrate how we can controlled the visible region  and the projection of the current density. The underlying idea is to chose voltages potential $f$ at the boundary to induced specific level curves $u_0 =$const. By doing so, we get a better understanding of the required direction of the current density to obtain an stable reconstruction. Here, $\delta \J_0$ denote $\J(\sigma) - \J(\sigma_0)$  and  $\v w_0 = \Pi_{\grad (u_0)}(\delta \J_0) $.}
  	\label{fig:controlled-stability-from-boundary}
\end{figure}

\begin{theorem}\label{theorem:controlled-projection}
Let $\sigma\in C^{1,\alpha}(\overline{\Omega})$,  $0<\alpha<1$, be positive in $\overline{\Omega}$ and $u$ be $\sigma$-harmonic with $|\nabla u|>0$ in $\overline\Omega$. There exists an $\epsilon>0$ depending on $\Omega$ and some  $C > 0$ depending on $\epsilon$, such that the following holds: If $\tsigma\in C^{1,\alpha}(\overline{\Omega})$ with
\begin{equation}\label{samesigmaOnGamma}
\|\sigma -\tsigma \|_{C^{1,\alpha}(\overline{\Omega})} < \epsilon\quad\mbox{and}\quad\sigma|_{\partial \Omega} =\tsigma|_{\partial \Omega},
\end{equation}
and $\tu$ is the $\tsigma$-harmonic map with
\begin{equation}\label{sameBC}
 \tu|_{\partial\Omega}= u |_{\partial \Omega},
\end{equation}
then

\begin{equation} \label{eq:satbility-controlled-projection}
\| \sigma - \tsigma\|_{L^2(\Omega)} \leq C \left|\left| \grad \cdot(\Pi_{\grad u} (\J(\sigma) - \J(\tsigma) ) \right|\right|^{\frac{\alpha}{2+\alpha}}_{L^2(\Omega)}.
\end{equation}
In particular,
\[
\| \sigma - \tsigma \|_{L^2(\Omega)} \leq C \left|\left| \Pi_{\grad u}(\J(\sigma) - \J(\tsigma)) \right|\right|^{\frac{\alpha}{2+\alpha}}_{H^1(\Omega)}.
\]
\end{theorem}

As a consequence of Theorem \ref{theorem:controlled-projection} let us consider the problem of determining a perturbation from a constant conductivity: If the boundary voltage is uniform in the $z$-direction, the equipotential surfaces are planes perpendicular to the $z$-direction and we only need two components of the magnetic field $\B = (B^x, B^y, B^z)$ to stably recover the conductivity, as exemplified below.

\begin{example}For simplicity assume that $\sigma \equiv1$ and $f(x,y,z) = z$. Then $u(x,y,z) = z$ and, by Amp\`ere's law,
\[
\delta J^z= \grad u\cdot \delta \J = \frac{1}{\mu_0} \grad u \cdot (\grad \times \delta  \B) = \frac{1}{\mu_0} \left( \frac{\partial} {\partial x}\delta B^y -
\frac{\partial }{\partial y}\delta B^x \right).
\]
By Theorem \ref{theorem:controlled-projection}, we can stably recover a sufficiently small pertubation of the conductivity from $\delta B^x$ and $\delta B^y$.
\end{example}


\section{Preliminaries}\label{preliminaries}
The stability result relies on the following decomposition of the perturbation $\delta\J$ in the data.

\begin{proposition}\label{prop:decomp-non-linear-problem}
Let $u$ and $\tilde{u}$ be $\sigma$-harmonic and $\tsigma$-harmonic, respectively. If $\grad(u + \tu) \neq 0$ in $\overline{V}$, for some open subset $V$ of $\Omega$, then
\begin{equation} \label{eq:decomp-non-linear-problem}
2 \grad \cdot \Pi_{\grad(u + \tu)} (\J(\sigma) - \J(\tsigma))
	=  L(u - \tu) \quad \mbox{in }V.
\end{equation}
where $L$ is the differential operator defined by
\begin{equation} \label{eq:L-operator}
Lv := -\nabla\cdot (\sigma+\tsigma) \nabla v + \nabla\cdot\left( (\sigma + \tsigma) \frac{\grad(u+\tu) \cdot\nabla v}{|\grad(u + \tu)|^2} \grad(u + \tu)\right).
\end{equation}
\end{proposition}
\begin{proof}
Consider the difference of two internal measurements
\begin{equation} \label{Diff-of-Current-Density}
2(\J(\sigma)  - \J(\tsigma))
	= (\sigma - \tsigma) \nabla (u + \tu) + (\sigma + \tsigma) \nabla (u - \tu).
\end{equation}
If we dot product with $\grad (u + \tu)$ we obtain
\begin{equation} \label{Projection-Diff-of-Current-Density}
2(\J(\sigma)  - \J(\tsigma))\cdot \grad (u +\tu)
	= (\sigma - \tsigma) |\grad(u + \tu)|^2 + (\sigma + \tsigma) \nabla (u - \tu)\cdot \nabla (u + \tu).
\end{equation}
Solving for $\sigma - \tsigma$ we have
\begin{equation}\label{Diff-of-Sigmas}
\sigma - \tsigma
	= \frac{2(\J(\sigma)  - \J(\tsigma))\cdot \grad (u +\tu) }{|\grad(u + \tu)|^2} - (\sigma + \tsigma) \frac{\grad (u + \tu) \cdot \grad (u - \tu)}{|\grad(u + \tu)|^2}
\end{equation}
On the other hand it follows easily from \eqref{Dirchlet-Problem}
\begin{equation} \label{Dirch-Prob-Sum}
\nabla \cdot (\sigma + \tsigma) \nabla(u - \tu)
	= - \nabla \cdot (\sigma - \tsigma) \nabla (u +\tu) \quad \mbox{in } \Omega.
\end{equation}
Using \eqref{Diff-of-Sigmas} we substitute $\sigma - \tsigma$ in \eqref{Dirch-Prob-Sum}  and we get
\begin{equation}\nonumber
\begin{split}
L (u-\tu)
		& =   - \grad \cdot (\sigma + \tsigma) \grad (u - \tu) \\
 &\quad+ \grad \cdot (\sigma + \tsigma) \frac{\grad (u + \tu) \cdot \grad (u - \tu)}{|\grad(u + \tu)|^2}	\grad (u+\tu) \\
 		& = \grad \cdot \left( \frac{2(\J(\sigma)  - \J(\tsigma))\cdot \grad (u +\tu) }{|\grad(u + \tu)|^2} \nabla (u + \tu) \right) \\
 		& = 2 \grad \cdot \Pi_{\grad(u + \tu)} (\J(\sigma) - \J(\tsigma))
\end{split}
\end{equation}
\end{proof}

We remark that \eqref{eq:decomp-non-linear-problem} is the non-linear version of the decomposition obtained in the linearized case in \cite{montalto2013stability}.
The representation of the operator \nolinebreak $L$ in local coordinates can then be obtained similarly as in \cite{montalto2013stability}:

\begin{proposition}\label{Prop:local_coordinates_L}
Let $u, \tu \in C^2(\overline{\Omega})$, with $\nabla (u+\tu)(x_0) \neq 0$ for $x_0 \in \Omega$. Then locally near $\Sigma_{x_0}$, the operator $L$ in \eqref{eq:L-operator} is the restriction of $\grad(\sigma +\tsigma)\grad$ onto the level surfaces $(u + \tu) =$const. Moreover, if $y^n = u+\tu$, and $y' = (y^1, \ldots, y^{n-1})$ are any local coordinates of the level set $\Sigma_{x_0}$, then in a neighborhood of $\Sigma_{x_0}$ the Euclidean line element $\d s$ is given by
\begin{equation} \label{LocalCoordLevelCurve}
 \d s^2  = c^2(\d y^n)^2 +  g_{\alpha\beta}\d y^\alpha \d y^\beta,
\quad g_{\alpha, \beta} : = \sum_i^{n-1} \frac{\partial x^i}{\partial y^\alpha}  \frac{\partial x^i}{\partial y^\beta} ,
\end{equation}
where $c= |\nabla (u+ \tu)|^{-1}$. In this coordinates the operator $L$ becomes
\begin{equation} \label{local_representation_L_not_explicit}
L = - \sum_{\alpha,\beta}\frac{1}{\sqrt{\det g}} \frac{\partial}{\partial y^\beta} (\sigma + \tsigma)  g^{\alpha\beta}  \sqrt{\det g} \frac{\partial}{\partial y^\alpha},
\end{equation}
where the Greek super/subscripts run from 1 to $n-1$.
\end{proposition}

\begin{proof}
Let $\Sigma_{x_0}$ be the level set of $u+\tu$ passing through $x_0$. On $\Sigma_{x_0}$ choose local coordinates $y'=(y^1,\ldots , y^{n-1})$ and set $y^n= (u+\tu) (x)$. Then $y=(y',y^n) = \varphi(x)$ define local coordinates in a neighborgood of $\Sigma_{x_0}$. In  the new coordinates, the Euclidean metric becomes
\begin{equation}\label{eq:metric-change-coordinates}
\d s^2 = c^2 (\d y^n)^2 + g_{\alpha\beta}\d y^\alpha \d y^\beta,\;\mbox{where}\; g_{\alpha\beta} : = \sum_{i=1}^{n-1} \frac{\partial x^i}{\partial y^\alpha}  \frac{\partial x^i}{\partial y^\beta},
\end{equation}
where $c= |\nabla (u+\tu)|^{-1}$. This follows easily by noticing that $u + \tu$ trivially solves the eikonal equation $c^2|\nabla \phi|^2=1$ for the speed $c= |\nabla (u+\tu)|^{-1}$. Then, near $x_0$, $u+\tu$ is the signed distance from the $x$ to the level surface $\Sigma_{x_0}$ and the Euclidean metric has block structure given by \eqref{eq:metric-change-coordinates}. Denote by $g$ be  metric after the change of variables, i.e.,  
\[
g = \left[ \begin{array}{cc}
 [g_{\alpha,\beta}] & 0 \\ 
0 & c^2
\end{array}  \right]
\]

Let $\phi\in C_0^\infty(\Omega)$ be a test function supported near some $\hat{x}_0 \in \Sigma_{x_0}$. We have
\begin{equation} \label{eq:L_form}
\begin{split}
\langle Lv,\phi\rangle
	& = \langle(\sigma + \tsigma)\nabla v,\nabla \phi\rangle - \langle(\sigma + \tsigma)\Pi_{\grad(u+\tu)}\nabla v,\nabla \phi \rangle,\\
	& = \left\langle (\sigma + \tsigma)\Pi^\perp_{\grad(u+\tu)}\nabla  v,\Pi^\perp_{\grad(u+\tu)} \nabla  \phi \right\rangle,
\end{split}
\end{equation}where $\langle\cdot,\cdot\rangle$ denotes the scalar product in $L^2(\Omega)$.

From \eqref{eq:metric-change-coordinates}, we get that for any function $v \in C^1(\overline{\Omega})$ 
\[
\Pi_{\grad(u+\tu)}\nabla_xv =  \left(0, \ldots ,0 , c^2 \frac{\partial v}{\partial y^n} \right).
\]
and 
\[
\Pi^{\perp}_{\grad(u+\tu)}\nabla_xv = \left(\frac{\partial v}{\partial y^1}, \ldots ,\frac{\partial v}{\partial y^{n-1}}, 0 \right).
\]
thus in $(y', y^n)$ coordinates \eqref{eq:L_form} becomes
\begin{equation}
\label{local_representation_L_form}
\begin{split}
\langle L  v,  \phi \rangle
	& =  \sum_{\alpha, \beta} \int_{\varphi^{-1}(\Omega)}(\sigma +\tsigma) \left(  g^{\alpha\beta}   \frac{\partial v}{\partial y^\alpha}    \frac{\partial \phi}{\partial y^\beta}  \right)\sqrt{\det g}\,\d y\\
& =- \sum_{\alpha, \beta} \int_{\varphi^{-1}(\Omega)}\frac{1}{\sqrt{\det g}} \left( \frac{\partial}{\partial y^\beta} (\sigma + \tsigma)  g^{\alpha\beta}  \sqrt{\det g} \frac{\partial}{\partial y^\alpha} \right)\phi(y)dy
\end{split}
\end{equation}
Since $\phi$ and $\hat{x}_0 \in \Sigma_{x_0}$ were arbitrarily  we conclude that
\begin{equation}\label{local_representation_L}
L = -\sum_{\alpha, \beta} \frac{1}{\sqrt{\det g}} \left( \frac{\partial}{\partial y^\beta} (\sigma + \tsigma)  g^{\alpha\beta}  \sqrt{\det g} \frac{\partial}{\partial y^\alpha} \right).
\end{equation}
near $\Sigma_0$.
\end{proof}

We note that  Proposition \ref{Prop:local_coordinates_L} shows  $L$  to be an elliptic differential operator acting onto the level sets of $u+\tu$, for which the normal component becomes a parameter.

\section{Stability estimates for partial data}\label{sec:proofofMainResult}

In this section we prove Theorem  \ref{theor:global-stability-CDII} and its corollary. The proof is based on establishing separate  estimates for $u-\tu$ and for the operator $L$ in \eqref{eq:decomp-non-linear-problem}. For points in $\S(\Gamma',u+\tu)$ the visibility from $\Gamma'$ parallel to $\grad(u + \tu)$ will allow us to control $u-\tu$, while the the visibility along the equipotential sets will be used for estimates on $L$.

Recall that $\S(\Gamma',u+\tu) = \S^+(\Gamma',u+\tu) \cup \S^-(\Gamma',u+\tu)$, where
\[
{\mathcal S^\pm}(\Gamma', u) = \{p \in \overline{\Omega}: \Sigma_p \cap \partial \Omega \subset \Gamma'  \mbox{ and } l^\pm_p \subset \I(\Gamma,u) \}.
\]We will show that
\begin{equation}\label{eq:divided-estimate}
\| \sigma - \tsigma \|_{L^2(\S^\pm(\Gamma', u+\tu))} \leq C \left|\left| \grad \cdot(\Pi_{\grad (u + \tu)} (\J(\sigma) - \J(\tsigma) ) ) \right|\right|^{\frac{\alpha}{2+\alpha}}_{L^2(\I(\Gamma', u+\tu))}.
\end{equation}
We prove \eqref{eq:divided-estimate} for $\S^+(\Gamma',u+\tu)$. The corresponding inequality for $\S^-(\Gamma',u+\tu)$ follows similarly.


To simplify notation,  let
\begin{equation}\label{notations}
\S  := \S^+(\Gamma,u+\tu),\quad \S': = \S^+(\Gamma ', u+\tu), \mbox{and }\; \I' := \I(\Gamma', u+\tu).
\end{equation}
We assume w.l.o.g. that $\S'$ and $\S$ are connected. Notice that, since there are only finitely many components, we can add the estimates that we obtain for the single-component case, to the more general case.

We will make use of the following two technical estimates in the Lemmas below.

\begin{lemma} \label{basicEstimates}There exist $C>0$ dependent on the domain $\Omega$ such that
\begin{equation} \label{StabilityEstimateForPotential}
\|\sigma - \tilde{\sigma}\|_{L^2(\S')} \leq C  \| u - \tu \|_{H^2(\S')},
\end{equation}
where for simplicity we denote $\S' := \S^+(\Gamma',u+\tu)$.
\end{lemma}

\begin{proof}Since $\nabla(u+\tu)\neq 0$ in $\overline\Omega$, without loss of generality (otherwise work with a diffeormorphic image of $\Omega$, and, possibly, with $-(u+\tu)$), we may assume that
\begin{align}\label{assumption}
\min_{\overline\Omega}\frac{\partial (u+\tu)}{\partial{x^n}}> 0.
\end{align}
Since the $n$-th coordinate now plays a special role, to simplify notation, let
$x':=(x^1,\ldots ,x^{n-1})$, so that $x=(x',x^n)$.

Let $Z$ be the set (possibly empty) of boundary points , where $\nabla(u+\tu)$ is tangential to the boundary,
\begin{align}\label{exceptions}
Z:=\{x\in\partial\Omega:\; n(x)\cdot\nabla (u+\tu)(x)=0\}.
\end{align}
The regularity assumptions on the boundary of the domain, and on $u,\tu$ yield that $Z$ is closed and confined to a co-dimension 1 variety of the boundary, in particular, $Z$ is negligible with respect to the induced area measure on the boundary. 

Let $\Gamma'\subset\subset\Gamma\setminus Z$, $x_0 \in {\S'}={ \S^+(\Gamma ', u+\tu)}$, and $\Sigma_0:= \{x\in\Omega:\; (u+\tu)(x)= (u+\tu)(x_0)\}$ be the level set passing through $x_0$.

Since $\frac{\partial(u+\tu)}{\partial x^n}(x_0',x^n_0)\neq0$, by the implicit function theorem, in a neighborhood
$O\subset\R^n$ of $x_0'$ there is a local representation of $\Sigma_0$:
$O\ni x'\mapsto (x',g(x'))$, for some $C^{1,\alpha}$- smooth coordinate map $g:O\subset \R^{n-1}\to\R$ with $g(x_0')=x^n_0$.

We define next a tubular neighborhood  $T_{x_0}$ by flowing along $\nabla(u+\tu)$ the points on $O\subset\Sigma_0$ until the boundary is met. More precisely, we consider the family (indexed in $x'\in O$) of the initial value problems
\begin{align}\label{IVP}
\frac{d}{dt}{\gamma}(t;x')=\frac{\nabla(u+\tu)}{|\nabla(u+\tu)|}(\gamma(t;x')),\quad \gamma(0)=(x',g(x')), \;x'\in O.
\end{align}
Since $x_0\in \S'$, for the solution $\gamma(\cdot;x_0')$ of \eqref{IVP} with parameter $x'=x_0'$ meets the boundary transversally, hence there exists a smallest time $t_{x_0'}$ such that $\gamma(t_{x_0'};x_0')\in \Gamma'$, and $\gamma(t;x_0')\in \Omega$ for $0\leq t<t_{x_0'}$.

The continuous dependence with parameters of solutions of initial value problems for ODE's shows that, by possibly shrinking the neighborhood $O$, we have that $t\mapsto\gamma(t,x')$ are defined on a maximal interval with $\gamma(\cdot,x'):[0,t_{x'})\to\Omega$ and $\gamma(t_{x'};x')\in\Gamma'$.

We consider the following change of coordinates $(x',x^n)\to (y',y^n)$, where $y'=x'$ and $x^n=\gamma(y^n;y')$. In the new coordinates, the tubular neighborhood $T_{x_0}$ rewrites 
\[
T_{x_0} = \{(y',y^n) \in \R^n: y'\in O,\, 0\leq y^n \leq t_{y'} \}.
\]
Let $J(y',y^n)$ denotes the absolute value of the Jacobian determinant of the change of coordinates above at some $(y',y^n)\in T_{x_0}$, and let $\beta>0$  denote an upper bound on the quotient

\begin{align}\label{jacobQuotient}
\beta:= \frac{\max_{T_{x_0}} J}{\min_{T_{x_0}}J} < \infty
\end{align}

Recall the equation \eqref{Dirchlet-Problem} interpreted as a transport equation for $\delta\sigma = (\sigma - \tsigma)$:
\begin{equation} \label{ODEforhOmega}
\nabla(u +\tu)\cdot \nabla (\delta\sigma) + (\delta\sigma) \Delta (u + \tu)  = G \quad \mbox{in } \Omega,
\end{equation}where, for brevity, we let
\begin{equation}
G:= - \nabla \cdot (\sigma + \tsigma) \nabla(u - \tu).
\end{equation}


In the local coordinates in the interior of $T_{x_0}$, the equation \eqref{ODEforhOmega} becomes
\begin{equation} \label{ODEforhLocal}
\frac{1}{c^2}\frac{\partial (\delta\sigma)}{\partial y^n} + (\delta\sigma)\Delta(u + \tu) = G,
\end{equation}
where $c = |\nabla (u+\tu)|^{-1}$. By \eqref{nonsingularity},
\begin{equation}\label{max_c}
\max_{\overline \Omega}c<\infty.
\end{equation}

By hypothesis \eqref{samesigmaOnGamma} we also have that
\begin{align}\label{h=0atBdd}
(\delta \sigma)(y', t_{y'}) =0,\quad\forall y'\in O.\end{align}

Using an integral factor we can solve \eqref{ODEforhLocal} and \eqref{h=0atBdd} to get
\begin{equation}
(\delta \sigma)(y',y^n) = \left( \frac{1}{\mu(y', y^n)} \int_{t_{y'}}^{y^n} c^2 G(y',t) \mu(y',t)dt \right),
\end{equation}where
\[
\mu(y',y^n) = \exp\left(-\int^{t_{y'}}_{y^n}c^2 \Delta(u+\tu)(y',t) dt \right).
\]

Let $m,M$ (not necessarily positive) be such that
\begin{equation}
m\leq c^2\Delta (u+\tu)(x)\leq M, \quad\forall x\in\overline\Omega.
\end{equation}Then, it is easy to see that, for $y'\in O$ and  $0\leq t\leq y^n\leq t_{y'}$, we have

\begin{equation}\label{boundMU}
e^{-|M|diam(\Omega)}\leq\frac{\mu(y',t)}{\mu(y',y^n)}\leq e^{|m|diam(\Omega)}.
\end{equation}


For a constant $\tilde C>0$ depending on the bounds in \eqref{boundMU} and \eqref{max_c}, and $\beta$ in \eqref{jacobQuotient}, we estimate
\begin{equation} \label{ineq:transport_lemma_estimate}
\begin{split}
\| \delta\sigma\|^2_{L^2(T_{x_0})} & = \int_{y'\in O}\int_0^{t_{y'}}  \left| \int_0^{y^n} c^2 G(y',t)\frac{\mu(y',t)}{\mu(y', y^n)} dt\right|^2  J(y',y^n) \d y^n\d y'\\
&\leq \tilde{C} \int_{y'\in O}\int_0^{t_{y'}} \int_0^{y^n} | G(y',t)|^2 J(y',y^n) \d t  \d y^n \d y'\\
&= \tilde{C} \int_{y'\in O} \int_0^{t_{y'}}\int_t^{t_y'} | G(y',t)|^2 \frac{J(y',y^n)}{J(y',t)}dy^n J(y',t) \d t  \d y'\\
&\leq \tilde{C}\beta \int_{y'\in O} \int_0^{t_{y'}}(t_y'-t)| G(y',t)|^2 dy^n J(y',t) \d t  \d y'\\
&\leq \tilde{C}\beta diam(\Omega)\int_{y'\in O} \int_0^{t_{y'}}| G(y',t)|^2 dy^n J(y',t) \d t  \d y'\\
& = \tilde{C}\beta diam(\Omega)\|G\|_{L^2(T_{x_0})}\\& 
\leq  \tilde{C}\beta diam(\Omega) \| \nabla \cdot (\sigma + \tsigma) \nabla (u - \tu) \|_{L^2(\S')},
\end{split}
\end{equation}

Since $\overline{\S'}$ can be covered by finitely many tubular neighborhoods, and the norm in the right hand side of \eqref{ineq:transport_lemma_estimate}, we conclude that
\begin{equation} 
\| \delta\sigma\|^2_{L^2(\S')} \leq C  \| \nabla \cdot (\sigma + \tsigma) \nabla (u - \tu) \|_{L^2(\S')} \leq C \| u - \tu \|_{H^2(\S')},
\end{equation}for a constant $C$ depending on $u,\tu$ and the domain $\Omega$, and a priori bounds on $C^1$-norm of $\sigma$ and $\tsigma$.
This finishes the proof of \eqref{StabilityEstimateForPotential}.

\end{proof}

\begin{lemma}\label{basicEstimate2}There exists $C>0$ dependent on $\Omega$ and a lower bound on $\sigma+\tsigma$, such that,
\begin{equation} \label{StabilityEstimateForL}
\|u - \tilde{u}\|_{L^2(\I')} \leq C \| L(u - \tilde{u})\|_{L^2(\I')},
\end{equation}where, for simplicity we denote by  $\I':=\I(\Gamma',u+\tu)$.
\end{lemma}
\begin{proof}

Let $x_0 $ be arbitrarily fixed in the injectivity region $\overline{\I'}$. Assume first that the level set $\Sigma_0$  of $u+\tu$ passing through $x_0$ intersects the boundary in $\Gamma'$, transversally at any point of intersection. Consider the local coordinates $(y',y^n)$ for the equipotential set $\Sigma_0$, introduced in Proposition \ref{Prop:local_coordinates_L}. By continuity of $y^n = u+\tu$ w.r.t. $y'$, for sufficiently small $\epsilon$, the  neighborhood $U_\epsilon$ of $\Sigma_0$ defined by 
\[
U_\epsilon := \{(y',y^n)\in \I': y'\in \Sigma_0, \mbox{and } y^n \in (-\epsilon,\epsilon) \}
\]
is \emph{visible from $\Gamma'$ along the equipotential sets of $u+\tu$}.

Using the local representation of $L$ obtained in (***) and local representation the metric $g$ on $U_\epsilon$ we estimate $L(\delta u)$. We denote by $g'$ the restriction of the metric to $y'$, i.e.,  $g' = [g_{\alpha,\beta}]$ for $1 \leq \alpha, \beta  \leq n-1$. Let $\delta u = u - \tu$. Note that in $U_\epsilon$, $(u - \tu)|_{\partial \Omega} =0$, hence Poincar\'e's inequality is valid in $\Sigma_0$. We estimate

\begin{align}
\langle L \delta u \cdot \delta u \rangle|_{U_\epsilon} 
& = \int\limits_{-\epsilon}^{\epsilon} \int\limits_{\Sigma_0} \left( \frac{1}{\sqrt{\det g}}  \frac{\partial  \delta u}{\partial y^\beta} (\sigma + \tsigma)  g^{\alpha\beta} \sqrt{\det g}
\frac{\partial \delta u}{\partial y^\alpha}\right) \cdot \delta u  \sqrt{\det g} \,\d y'\d y^n\nonumber\\
&=\int\limits_{-\epsilon}^{\epsilon} \int\limits_{\Sigma_0} (\sigma + \tsigma)  g^{\alpha\beta}
\frac{\partial \delta u}{\partial y^\alpha} \frac{\partial  \delta u}{\partial y^\beta} \sqrt{\det g} \; \d y'\d y^n\nonumber\\
&\ge \delta\int_{-\epsilon}^\epsilon \int_{\Sigma_0} |\nabla_{y'} \delta u |^2 \; \; \d y'\d y^n,\nonumber\\
&\ge \delta\int_{-\epsilon}^\epsilon \int_{\Sigma_0} |\delta u|^2 \;  \;\d y'\d y^n\nonumber\\
&=\delta\|\delta u\|^2_{L^2(U_\epsilon)},
\label{L-Form}
\end{align}
where $\delta>0$ depends on an a priori lower bound of $\sigma + \tsigma$ and of $g^{\alpha,\beta}$ on $\overline\Omega$ for $1 \leq \alpha,\beta \leq n$. The second inequality is the Poincar\'e inequality. Since $U_\epsilon \subset \I'$ we obtain from \eqref{L-Form} 
\begin{align} \label{eq:integration-by-parts}
\|\delta u\|^2_{L^2(U_\epsilon)} & \leq C \langle L(\delta u), \delta u\rangle_{L^2(\I')},
\end{align}
By compactness, we can get a finite covering of $\overline{\I'}$ with neighborhoods $U_\epsilon$. Summing \eqref{eq:integration-by-parts} over the covering of $U_\epsilon$ we obtain 
\begin{equation}
\|u - \tu \|^2_{L^2(\I')} \leq C (L(u - \tu), u - \tu)_{L^2(\I')},
\end{equation}
for some constant $C$ which depends on $\Omega$, and the lower bound on $\sigma+\tsigma$. Finally from Cauchy-Schwartz inequality \eqref{StabilityEstimateForL} follows.

Finally the level set $\Sigma_0$ intersects tangentially the boundary of $\Omega$ as some point then the
analysis presented above could fail depending on the degeneracy of such interseccion. In particular it is not clear how to justify integration by parts on $U_\epsilon$ in the intersection is degenerate. To avoid this difficulty we will extend $\sigma, \tsigma, u$ and $\tu$ in two different ways to an open domain $\Omega_1$ containing $\overline{\I'}$. The first extension takes advantage of the boundary information while the second extension uses the regularity of the functions to extend the geometry and the operators in decomposition \eqref{eq:decomp-non-linear-problem}.

Since $\Gamma'$ is compactly supported in $\Gamma$ it follows that $\overline{\I'} \subset \I$. Thus there exists an open bounded $\Omega_1$ with smooth boundary, such that $\overline{\I'} \subset \Omega_1$ and $\Omega_1 \cap \Omega \subset \I$. First, using the fact that $(\sigma - \tsigma)|_{\Gamma} =0$ we can find $H^1(\Omega_1)$ extensions $\sigma_1$ and $\tsigma_1$ of $\sigma$ and $\tsigma$, respectively, such that $(\sigma_1 - \tsigma_1)=0$ in ${\Omega_1\setminus \Omega}$; (e.g., extend first $\sigma$ to some $\sigma_1\in C^2(\Omega_1)$, and define
$\tsigma_1 := \sigma_1$ in $\Omega_1 \setminus \Omega$. Since $\tsigma_1$ coincide with $\sigma_1$ on $\partial\Omega$, the difference $\tsigma_1-\sigma_1$, and thus $\tsigma_1$ lie in $ H^1(\Omega_1)$. Moreover, since $\sigma$ and $\tsigma$ are positive in $\overline{\Omega}$, we may assume that $\sigma_1$ and $\tsigma_1$ are positive in $\Omega_1$. Second, denote by $\sigma_2, \tsigma_2 \in C^2(\Omega_1)$ extensions of $\sigma$ and $\tsigma$, respectively. We can take $\Omega_1$ close enough to the boundary of $\Omega$ so that $\sigma_2$ and $\tsigma_2$ are positive in $\overline{\Omega_1}$.

In a similar way we extend $u$ and $\tu$ as follows. Since $(u -\tu)|_{\Gamma} = 0$, we can find $H^1(\Omega_1)$ extensions $u_1$ and $\tu_1$ of $u$ and $\tu$, respectively, such that $(u_1 - \tu_1)|_{\Omega_1\setminus \Omega} = 0$. Also denote by $u_2, \tu_2 \in C^2(\Omega_1)$ extensions of $u$ and $\tu$ respectively. We can take $\Omega_1$ close enough to the boundary of $\Omega$ so that $\nabla (u_2 + \tilde{u}_2) \neq 0$ on $\overline{\Omega_1}$. Finally we denote by $g_1$ a $C^1(\Omega_1)$ extension to $\Omega_1$ of the metric $g$ obtained in \eqref{LocalCoordLevelCurve}.

We then define $U_\epsilon$ in $\Omega_1$ instead of $\overline{\I'}$ by 
\[
U_\epsilon := \{(y',y^n)\in \Omega_1:  y'\in \Sigma_0, \mbox{and } y^n \in (-\epsilon,\epsilon) \}
\]
Now, even if the intersection of $\Sigma_0$ with $\partial \Omega_1$ is tangential the function that we are integrating is compactly supported, so we can justify integration by parts in \eqref{eq:integration-by-parts}. 
\end{proof}

Using the estimates in Lemmas \ref{basicEstimates} and \ref{basicEstimate2} we proceed to proving \eqref{J-Projection-Estimate}.  Recall from Schauder theory that $u, \tu \in C^{2,\alpha}(\Omega)$. By using in order the estimate \eqref{StabilityEstimateForPotential}, an interpolation estimate in \cite[Sec. 4.3.1 ]{MR1328645}, \eqref{StabilityEstimateForL}, and Proposition \ref{prop:decomp-non-linear-problem} we obtain
\begin{equation} \label{J-Estimate-H1}
\begin{split}
 \| \sigma - \tsigma \|_{L^2(V')}
  		&	\leq C\| u - \tu \|_{H^2(V')}  \\
  		&	\leq C  \| u - \tu \|^{\frac{\alpha}{2 + \alpha}}_{L^2(V')} \cdot\| u - \tu \|^{\frac{2}{2+\alpha}}_{H^{2+\alpha}(V')}
  		\leq C \|u - \tilde{u}\|^{\frac{\alpha}{2 + \alpha}}_{L^2(\I')} \\
  		&	\leq C \| L(u - \tu)\|^{\frac{\alpha}{2 + \alpha}}_{L^2(\I')} \\
  		&	=    C  \left|\left|2 \grad \cdot\Pi_{\nabla(u+\tu)} (\J(\sigma)  - \J(\tsigma)) \right|\right|^{\frac{\alpha}{2 + \alpha}}_{L^2(\I')} \\
\end{split}
\end{equation}
This finishes the proof of Theorem \ref{theor:global-stability-CDII}.

\begin{proof}[Proof of Theorem \ref{theorem:controlled-projection}]
Recall $\epsilon>0$ in \eqref{samesigmaOnGamma}.
The Schauder stability estimates for solutions of elliptic equations with $C^{\alpha}(\Omega)$-coefficients and $C^{2,\alpha}$ traces on the boundary, yield
\begin{align}\label{schauder}
\|u-\tu\|_{C^{2,\alpha}(\overline\Omega)}\leq M\epsilon,
\end{align}for a constant $M$ dependent only on $\Omega$ and a lower bound of $\sigma$.
In particular, for $\epsilon<1/M$ we have that 
\begin{align}\label{C1boundOfSUM}
\|\tu\|_{C^2(\overline\Omega)}\leq \|u\|_{C^2(\overline\Omega)}+1.
\end{align}Moreover, since $m:=\min_{\overline\Omega}|\nabla u|>0$, for $\epsilon<m/M$, we also have
\begin{align}\label{C1lowerBDsum}
|\nabla (u+\tu)|\geq 2|\nabla u|-|\nabla(u-\tu)|\geq 2m-M\epsilon \geq m>0.
\end{align}

Denote by $\delta \J = \J(\sigma) - \J(\tsigma)$. From Theorem \ref{theor:global-stability-CDII} we have that there exist a $C>0$ \emph{independent of $\epsilon$}, such that
\begin{equation}\label{eq:proof-2.1}
\| \sigma - \tsigma \|_{L^2(\Omega)} \leq C \left| \left| \frac{\delta \J \cdot \grad(u +\tu)}{|\grad(u +\tu)|} \right|\right|_{H^1(\Omega)}.
\end{equation}
Write
\begin{equation}\label{eq:proof-2.2}
 \frac{\delta \J \cdot \grad(u +\tu)}{|\grad(u +\tu)|} = 2\frac{|\grad u|}{|\grad(u +\tu)|}\left( \frac{\delta \J \cdot \grad u}{|\grad u|}\right) + \frac{\delta \J \cdot \grad(\tu-u)}{|\grad(u +\tu)|}.
\end{equation}
Note that all the terms in \eqref{eq:proof-2.2} are $C^1(\overline\Omega)$-regular.

By using the identity \eqref{Projection-Diff-of-Current-Density}, the second term in the right hand side of \eqref{eq:proof-2.2} rewrites as
\begin{equation}\label{eq:proof-2.3}
\frac{\delta \J \cdot \grad(\tu-u)}{|\grad(u +\tu)|} = \frac{1}{2} (\sigma -\tsigma)|\grad (u-\tu)| + \frac{1}{2}(\sigma +\tsigma)\frac{|\grad(u-\tu)|^2}{|\grad(u+\tu)|}.
\end{equation}
Using \eqref{Dirch-Prob-Sum}, \eqref{C1boundOfSUM}, and the Cauchy-Schwartz inequality we obtain
\begin{align}
\langle(\sigma + \tsigma)\nabla (\tu-u), \nabla(\tu -u)\rangle_{L^2(\Omega)}
	&	= - \langle\nabla \cdot (\sigma + \tsigma)\nabla (\tu-u), \tu -u\rangle_{L^2(\Omega)}\nonumber \\
 	&	= \langle\nabla \cdot (\tsigma - \sigma)\nabla (u +\tu), \tu -u\rangle_{L^2(\Omega)} \nonumber\\
 	&	\leq \|(\tsigma - \sigma)\grad (u+\tu)\|_{L^2(\Omega)} \|\grad(\tu-u)\|_{L^2(\Omega)} \nonumber \\
 	&	\leq  C \|\tsigma -\sigma \|_{L^2(\Omega)} \| \grad(\tu-u)\|_{L^2(\Omega)},\label{intermediate}
\end{align}
for $C$ dependent on the $C^1$-norm of $u$ but independent of $\epsilon$. From \eqref{intermediate}, for$$\delta :=\min_{\overline\Omega}(\sigma + \tsigma) >0,$$
we obtain
\begin{equation} \label{eq:proof-2.35}
\|\grad(u-\tu)\|_{L^2(\Omega)} \leq \frac{C}{\delta} \|\sigma -\tsigma \|_{L^2(\Omega)}.
\end{equation}

We claim that 
\begin{equation}\label{eq:proof-2.4}
\left|\left|\frac{\delta \J \cdot \grad(u-\tu)}{|\grad(u +\tu)|}\right|\right|_{H^1(\Omega)} \leq C \epsilon \| \sigma - \tsigma\|_{L^2(\Omega)}.
\end{equation}
for some $C$ independent of $\epsilon$. Indeed: one differentiation in \eqref{eq:proof-2.3} in any variable, which for simplicity we denote by subscript $x$, yields four terms. We treat each term individually, the appearing constants are independent of $\epsilon$. To bound the first term we use \eqref{samesigmaOnGamma}, the Cauchy inequality and \eqref{eq:proof-2.35} to obtain
\begin{align*}
\int_\Omega |(\sigma-\tsigma)_x|\,|\nabla (u-\tu)|dx\leq C_1\epsilon\|\sigma-\tsigma\|_{L^2(\Omega)}.
\end{align*}To bound the second term we use the Cauchy inequality and uniform bounds (dependent on $u$ only via \eqref{C1boundOfSUM}) on the second derivatives of $u$ and $\tu$ and obtain
\begin{align*}
\int_\Omega |\sigma-\tsigma|\,(|\nabla (u-\tu)|)_x dx\leq C_2\epsilon\|\sigma-\tsigma\|_{L^2(\Omega)}.
\end{align*}To bound the third term we use uniform bounds on the second derivatives of $u$ (and implicitly on $\tu$ via \eqref{C1boundOfSUM}) and on the first derivatives of $\sigma$ (and implicitly of $\tsigma$ via \eqref{samesigmaOnGamma}), \eqref{C1lowerBDsum} and \eqref{eq:proof-2.35} and obtain
\begin{align*}
\int_\Omega\left(\frac{\sigma+\tsigma}{|\nabla (u+\tu)|}\right)_x|\nabla (u-\tu)|^2dx\leq C_3\epsilon\|\sigma-\tsigma\|_{L^2(\Omega)}.
\end{align*}To bound the fourth term we use the Cauchy inequality and uniform bounds on $\sigma,\tsigma$, and $C^2$-bounds on $u$ (and implicitly via on $\tu$ \eqref{C1boundOfSUM}), \eqref{C1lowerBDsum}, and \eqref{eq:proof-2.35} and obtain
\begin{align*}
\int_\Omega\frac{\sigma+\tsigma}{|\nabla (u+\tu)|}(|\nabla(u-\tu)|)_x |\nabla (u-\tu)|dx\leq C_4\epsilon\|\sigma-\tsigma\|_{L^2(\Omega)}.
\end{align*}This finishes the proof of \eqref{eq:proof-2.4}.

From Equations \eqref{eq:proof-2.1}, \eqref{eq:proof-2.2} and \eqref{eq:proof-2.4} we get
\begin{equation} \label{eq:proof-2.6}
\| \sigma - \tsigma \|_{L^2(\Omega)} \leq C \left| \left| \frac{\delta \J \cdot \grad u}{|\grad u|} \right|\right|_{H^1(\Omega)} + C\epsilon \| \sigma - \tsigma\|_{L^2(\Omega)},
\end{equation}for a constant $C$ independent of $\epsilon$.

Finally, using \eqref{eq:proof-2.6} for $\epsilon < 1/C$ we obtain
\[
\| \sigma - \tsigma \|_{L^2(\Omega)} \leq \frac{C}{1-C\epsilon} \left| \left| \frac{\delta \J \cdot \grad u}{|\grad u|} \right|\right|_{H^1(\Omega)}
\]
which implies \eqref{eq:satbility-controlled-projection} and proves Theorem \ref{theorem:controlled-projection}.
\end{proof}

\begin{center}
\textbf{Acknowledgments}
\end{center}
We would express our gratitute to  Plamen Stefanov for his valuable suggestions during the writing of this paper. 

\bibliographystyle{amsplain}

\end{document}

%% file: sigma-p-and-lp.pdf_tex
\begingroup%
  \makeatletter%
  \providecommand\color[2][]{%
    \errmessage{(Inkscape) Color is used for the text in Inkscape, but the package 'color.sty' is not loaded}%
    \renewcommand\color[2][]{}%
  }%
  \providecommand\transparent[1]{%
    \errmessage{(Inkscape) Transparency is used (non-zero) for the text in Inkscape, but the package 'transparent.sty' is not loaded}%
    \renewcommand\transparent[1]{}%
  }%
  \providecommand\rotatebox[2]{#2}%
  \ifx\svgwidth\undefined%
    \setlength{\unitlength}{306.92080078bp}%
    \ifx\svgscale\undefined%
      \relax%
    \else%
      \setlength{\unitlength}{\unitlength * \real{\svgscale}}%
    \fi%
  \else%
    \setlength{\unitlength}{\svgwidth}%
  \fi%
  \global\let\svgwidth\undefined%
  \global\let\svgscale\undefined%
  \makeatother%
  \begin{picture}(1,0.69262639)%
    \put(0,0){\includegraphics[width=\unitlength]{sigma-p-and-lp.pdf}}%
    \put(0.40919292,0.28839812){\color[rgb]{0,0,0}\makebox(0,0)[lb]{\smash{$\footnotesize p$}}}%
    \put(0.82415812,0.26935517){\color[rgb]{0,0,0}\makebox(0,0)[lb]{\smash{$\footnotesize \Sigma_p$}}}%
    \put(0.51079157,0.48456967){\color[rgb]{0,0,0}\makebox(0,0)[lb]{\smash{$\footnotesize x_p(t)$}}}%
    \put(0.05026824,0.1119474){\color[rgb]{0,0,0}\makebox(0,0)[lb]{\smash{$\footnotesize \Omega$}}}%
    \put(0.16679831,0.68040061){\color[rgb]{0,0,0}\makebox(0,0)[lb]{\smash{$\footnotesize x(t^+_p)$}}}%
    \put(0.3788793,0.00329126){\color[rgb]{0,0,0}\makebox(0,0)[lb]{\smash{$\footnotesize x(t^-_p)$}}}%
    \put(0.36823274,0.54275059){\color[rgb]{0,0,0}\makebox(0,0)[lb]{\smash{$\footnotesize l_p^+$}}}%
    \put(0.46795492,0.14014949){\color[rgb]{0,0,0}\makebox(0,0)[lb]{\smash{$\footnotesize l_p^-$}}}%
  \end{picture}%
\endgroup%

%% file: visible-region-non-connected.pdf_tex
\begingroup%
  \makeatletter%
  \providecommand\color[2][]{%
    \errmessage{(Inkscape) Color is used for the text in Inkscape, but the package 'color.sty' is not loaded}%
    \renewcommand\color[2][]{}%
  }%
  \providecommand\transparent[1]{%
    \errmessage{(Inkscape) Transparency is used (non-zero) for the text in Inkscape, but the package 'transparent.sty' is not loaded}%
    \renewcommand\transparent[1]{}%
  }%
  \providecommand\rotatebox[2]{#2}%
  \ifx\svgwidth\undefined%
    \setlength{\unitlength}{319.28692342bp}%
    \ifx\svgscale\undefined%
      \relax%
    \else%
      \setlength{\unitlength}{\unitlength * \real{\svgscale}}%
    \fi%
  \else%
    \setlength{\unitlength}{\svgwidth}%
  \fi%
  \global\let\svgwidth\undefined%
  \global\let\svgscale\undefined%
  \makeatother%
  \begin{picture}(1,0.62114699)%
    \put(0,0){\includegraphics[width=\unitlength]{visible-region-non-connected.pdf}}%
    \put(0.92182503,0.48153224){\color[rgb]{0.50196078,0,0}\makebox(0,0)[lb]{\smash{$\Gamma$}}}%
    \put(-0.00042679,0.47638716){\color[rgb]{0.50196078,0,0}\makebox(0,0)[lb]{\smash{$\Gamma$}}}%
    \put(0.09064692,0.42195347){\color[rgb]{0,0,0}\makebox(0,0)[lb]{\smash{\footnotesize$S(\Gamma,u)$}}}%
    \put(0.46193138,0.41948709){\color[rgb]{0,0,0}\makebox(0,0)[lb]{\smash{\footnotesize$\I(\Gamma, u)$}}}%
    \put(0.73783248,0.00316379){\color[rgb]{0,0,0}\makebox(0,0)[lb]{\smash{$\Omega$}}}%
    \put(0.44815546,0.56284271){\color[rgb]{0,0,0}\makebox(0,0)[lb]{\smash{\footnotesize $u=$const.}}}%
    \put(0.45942697,0.24104157){\color[rgb]{0,0,0}\makebox(0,0)[lb]{\smash{\footnotesize $u=$const.}}}%
    \put(0.3200034,0.30289823){\color[rgb]{0,0,0}\makebox(0,0)[lb]{\smash{\footnotesize $\nabla u$}}}%
    \put(0.76448302,0.41918498){\color[rgb]{0,0,0}\makebox(0,0)[lb]{\smash{\footnotesize$S(\Gamma,u)$}}}%
  \end{picture}%
\endgroup%

%% file: stability-region-connected.pdf_tex
\begingroup%
  \makeatletter%
  \providecommand\color[2][]{%
    \errmessage{(Inkscape) Color is used for the text in Inkscape, but the package 'color.sty' is not loaded}%
    \renewcommand\color[2][]{}%
  }%
  \providecommand\transparent[1]{%
    \errmessage{(Inkscape) Transparency is used (non-zero) for the text in Inkscape, but the package 'transparent.sty' is not loaded}%
    \renewcommand\transparent[1]{}%
  }%
  \providecommand\rotatebox[2]{#2}%
  \ifx\svgwidth\undefined%
    \setlength{\unitlength}{166.09724121bp}%
    \ifx\svgscale\undefined%
      \relax%
    \else%
      \setlength{\unitlength}{\unitlength * \real{\svgscale}}%
    \fi%
  \else%
    \setlength{\unitlength}{\svgwidth}%
  \fi%
  \global\let\svgwidth\undefined%
  \global\let\svgscale\undefined%
  \makeatother%
  \begin{picture}(1,0.80996696)%
    \put(0,0){\includegraphics[width=\unitlength]{stability-region-connected.pdf}}%
    \put(-0.00103008,0.33406661){\color[rgb]{0.50196078,0,0}\makebox(0,0)[lb]{\smash{$\Gamma$}}}%
    \put(0.8635193,0.54921393){\color[rgb]{0,0,0}\makebox(0,0)[lb]{\smash{$\Omega$}}}%
    \put(0.2780899,0.37046327){\color[rgb]{0,0,0}\makebox(0,0)[lb]{\smash{\footnotesize $\nabla u$}}}%
    \put(0.46570429,0.37500457){\color[rgb]{0,0,0}\makebox(0,0)[lb]{\smash{\footnotesize $u =$const. }}}%
    \put(0.24048839,0.58231265){\color[rgb]{0,0,0}\makebox(0,0)[lb]{\smash{\footnotesize $u =$const. }}}%
    \put(0.3870872,0.22216148){\color[rgb]{0,0,0}\makebox(0,0)[lb]{\smash{\footnotesize $\I(\Gamma, u) = \S(\Gamma, u)$  }}}%
    \put(0.14005923,0.69952591){\color[rgb]{0,0,0}\makebox(0,0)[lb]{\smash{\footnotesize $\I(\Gamma, u) = \S(\Gamma, u)$  }}}%
  \end{picture}%
\endgroup%

%% file: projection.pdf_tex
\begingroup%
  \makeatletter%
  \providecommand\color[2][]{%
    \errmessage{(Inkscape) Color is used for the text in Inkscape, but the package 'color.sty' is not loaded}%
    \renewcommand\color[2][]{}%
  }%
  \providecommand\transparent[1]{%
    \errmessage{(Inkscape) Transparency is used (non-zero) for the text in Inkscape, but the package 'transparent.sty' is not loaded}%
    \renewcommand\transparent[1]{}%
  }%
  \providecommand\rotatebox[2]{#2}%
  \ifx\svgwidth\undefined%
    \setlength{\unitlength}{157.52849121bp}%
    \ifx\svgscale\undefined%
      \relax%
    \else%
      \setlength{\unitlength}{\unitlength * \real{\svgscale}}%
    \fi%
  \else%
    \setlength{\unitlength}{\svgwidth}%
  \fi%
  \global\let\svgwidth\undefined%
  \global\let\svgscale\undefined%
  \makeatother%
  \begin{picture}(1,0.94448896)%
    \put(0,0){\includegraphics[width=\unitlength]{projection.pdf}}%
    \put(0.3864154,0.61510063){\color[rgb]{0,0,0}\makebox(0,0)[lb]{\smash{\footnotesize $\nabla u$}}}%
    \put(0.30715534,0.00641253){\color[rgb]{0,0,0}\makebox(0,0)[lb]{\smash{\footnotesize $\nabla \tilde{u}$}}}%
    \put(0.52099439,0.25471257){\color[rgb]{0,0,0}\makebox(0,0)[lb]{\smash{\footnotesize $\nabla (u + \tilde{u})$}}}%
    \put(0.21592464,0.25321576){\color[rgb]{0,0,0}\makebox(0,0)[lb]{\smash{\footnotesize $\v w$}}}%
    \put(0.20195891,0.87791023){\color[rgb]{0,0,0}\makebox(0,0)[lb]{\smash{\footnotesize $\delta\mathbf{J}$}}}%
  \end{picture}%
\endgroup%

%% file: controlled-stability-from-boundary.pdf_tex
\begingroup%
  \makeatletter%
  \providecommand\color[2][]{%
    \errmessage{(Inkscape) Color is used for the text in Inkscape, but the package 'color.sty' is not loaded}%
    \renewcommand\color[2][]{}%
  }%
  \providecommand\transparent[1]{%
    \errmessage{(Inkscape) Transparency is used (non-zero) for the text in Inkscape, but the package 'transparent.sty' is not loaded}%
    \renewcommand\transparent[1]{}%
  }%
  \providecommand\rotatebox[2]{#2}%
  \ifx\svgwidth\undefined%
    \setlength{\unitlength}{203.40356445bp}%
    \ifx\svgscale\undefined%
      \relax%
    \else%
      \setlength{\unitlength}{\unitlength * \real{\svgscale}}%
    \fi%
  \else%
    \setlength{\unitlength}{\svgwidth}%
  \fi%
  \global\let\svgwidth\undefined%
  \global\let\svgscale\undefined%
  \makeatother%
  \begin{picture}(1,0.71287245)%
    \put(0,0){\includegraphics[width=\unitlength]{controlled-stability-from-boundary.pdf}}%
    \put(0.37004613,0.69442467){\color[rgb]{0.34117647,0.06666667,0}\makebox(0,0)[lb]{\smash{\footnotesize $\Gamma$}}}%
    \put(0.71064298,0.20887){\color[rgb]{0,0,0}\makebox(0,0)[lb]{\smash{\footnotesize $u_0 =$const. }}}%
    \put(0.38213796,0.41814803){\color[rgb]{0,0,0}\makebox(0,0)[lb]{\smash{\footnotesize $\nabla u_0$}}}%
    \put(0.14855937,0.05559103){\color[rgb]{0,0,0}\makebox(0,0)[lb]{\smash{$\Omega$}}}%
    \put(0.2325904,0.42698927){\color[rgb]{0,0,0}\makebox(0,0)[lb]{\smash{$\footnotesize \delta\mathbf{J}_0$}}}%
    \put(0.32951937,0.36333221){\color[rgb]{0.09803922,0.49803922,0.09803922}\makebox(0,0)[lb]{\smash{\footnotesize$\v w_0$}}}%
    \put(0,0.51568098){\color[rgb]{0.50196078,0,0}\makebox(0,0)[lb]{\smash{\footnotesize voltage $f$}}}%
    \put(0.61732108,0.36905982){\color[rgb]{0.50196078,0,0}\makebox(0,0)[lb]{\smash{\footnotesize voltage $f$}}}%
    \put(0.68343798,0.50025178){\color[rgb]{0,0,0}\makebox(0,0)[lb]{\smash{\footnotesize $u_0 =$const. }}}%
  \end{picture}%
\endgroup%